\RequirePackage{amsmath, amssymb, amsfonts}

\documentclass[runningheads]{llncs}
\usepackage[T1]{fontenc}
\usepackage{graphicx}
\usepackage[utf8]{inputenc}
\usepackage{xcolor}
\usepackage{overpic}
\usepackage{wrapfig,lipsum}
\usepackage{hyperref}
\usepackage{rotating}
\usepackage{booktabs}
\usepackage{nicefrac}

\newcommand{\ddd}{\scriptstyle{\textnormal{d}}}
\newcommand{\dd}{\textnormal{d}}
\newcommand{\grad}{\textnormal{grad}}
\newcommand{\GL}{\textnormal{GL}}

\newcommand{\SO}{\textnormal{SO}}
\newcommand{\SE}{\textnormal{SE}}

\newcommand{\fcdot}{\; \cdot\; }
\newcommand{\diff}{\mathrm{d}}
\newcommand{\R}{\mathbb{R}}

\newcommand{\Lt}{\mathrm{L}}
\newcommand{\Rt}{\mathrm{R}}
\newcommand{\Tt}{\mathrm{T}}

\DeclareMathOperator{\dist}{dist}
\DeclareMathOperator{\Log}{Log}
\DeclareMathOperator{\Exp}{Exp}

\DeclareMathOperator*{\argmin}{arg\,min}

\usepackage{amsthm}

\begin{document}
\title{Bi-invariant Geodesic Regression\\with Data from the Osteoarthritis Initiative}

\titlerunning{Bi-invariant Geodesic Regression}

\author{Johannes Schade\inst{1, 3}\orcidID{0009-0002-6207-4059} \and
Christoph von Tycowicz\inst{3}\orcidID{0000-0002-7120-4081} \and
Martin Hanik\inst{2, 3}\orcidID{0000-0002-1447-4069}}
\authorrunning{J. Schade et al.}
%
\institute{{Freie Universität Berlin, Berlin, Germany} \and
{Technical University Berlin, Berlin, Germany} \and
Zuse Institute Berlin, Berlin, Germany \\
\email{\{johannes.schade,vontycowicz,hanik\}@zib.de}}

\maketitle              
\begin{abstract}
Many phenomena are naturally characterized by measuring continuous transformations such as shape changes in medicine or articulated systems in robotics.
Modeling the variability in such datasets requires performing statistics on Lie groups, that is, manifolds carrying an additional group structure.
As the Lie group captures the symmetries in the data, it is essential from a theoretical and practical perspective to ask for statistical methods that respect these symmetries; this way they are insensitive to confounding effects, e.g., due to the choice of reference coordinate systems. 
In this work, we investigate geodesic regression---a generalization of linear regression originally derived for Riemannian manifolds.
While Lie groups can be endowed with Riemannian metrics, these are generally incompatible with the group structure.
We develop a non-metric estimator using an affine connection setting. It captures geodesic relationships respecting the symmetries given by left and right translations.
For its computation, we propose an efficient fixed point algorithm requiring simple differential expressions that can be calculated through automatic differentiation.
We perform experiments on a synthetic example and evaluate our method on an open-access, clinical dataset studying knee joint configurations under the progression of osteoarthritis.

\end{abstract}
\section{Introduction}

We begin with the essential mathematical background before introducing the subject of this work: regression for geometric data using geodesic models.
A discussion of related work and our contributions is deferred to the end of this introduction.

\subsubsection{Basics of Lie groups.}
    We briefly review Lie group theory; refer to~\cite{Postnikov2013} for more detail. In the following, ``smooth'' means ``infinitely differentiable''.
    
    A Lie group $G$ is a smooth manifold that has a compatible group structure; that is, it comes with a smooth, associative (often not commutative) group operation $G \times G \ni (g,h) \mapsto gh \in G$ with identity element $e$ and a smooth inversion map $G \ni g \mapsto g^{-1}$. 
    
    The group operation defines two automorphisms for each $g$ in a general Lie group $G$: the left and right translation $\Lt_h: g \mapsto hg$ and $\Rt_h: g \mapsto gh$. Their derivatives $d_g\Lt_h$ and $d_g\Rt_h$ map tangent vectors $v \in T_gG$ bijectively to the tangent spaces 
    $T_{hg}G$ and $T_{gh}G$ at $hg$ and $gh$, respectively. 
    
    Each $v \in T_eG$ defines a left-invariant vector field $X$ by $X_g = \dd_e\Lt_g(v)$ for all $g \in G$. Right invariant vector fields are constructed analogously. 
    The integral curve $\alpha_v: \mathbb{R} \to G$ of a (left or right) invariant vector field $X$ with $v = X_e$ determines a 1-parameter subgroup of $G$ through $e$. The \textit{group exponential} $\exp$ is then defined by $\exp(v):= \alpha_v(1)$; being a diffeomorphism around $e$, it has a local inverse: the \textit{group logarithm} log.
    
    Given two vector fields $X, Y$ on $G$, a so-called connection $\nabla$ enables differentiation of $Y$ along $X$, yielding a vector field $\nabla_X Y$. With $\gamma' := \frac{\ddd \gamma}{\ddd t}$, a \textit{geodesic} $\gamma: [0,1] \to G$ is defined by $\nabla_{\gamma'} \gamma' = 0$, creating a curve without acceleration. Every $f \in G$ has a so-called normal convex neighborhood $U$, in which any two points $f,h \in U$ can be connected by a unique geodesic $[0,1] \ni t \mapsto \gamma(t;f,h)$ that lies within $U$. With $v:= \gamma'(0;g,h)$, the \textit{connection exponential} $\Exp_g: T_gG \to G$ at $g$ is defined by $\Exp_g(v):= \gamma(1;g,h)$; it is also a local diffeomorphism with the \textit{connection logarithm} $\Log_g(h) = \gamma'(0;g,h)$ as its inverse. These maps are called \textit{Riemannian exponential} and \textit{logarithm}, respectively, if $\nabla$ is the Levi-Civita connection of a Riemannian metric.

    Translated group and Riemannian maps coincide if a bi-invariant Riemannian metric (that is, one that remains invariant under both left and right translations) is used. Unfortunately, many Lie groups do not allow for such a metric~\cite{MiolanePennec2015}.
    If, on the other hand, $\nabla$ is the canonical Cartan-Schouten (CCS) connection~\cite{CartanSchouten1926,Postnikov2013}, then geodesics and left (or right) translated 1-parameter subgroups \textit{always} coincide. 
    More precisely, there is a maximal normal convex neighborhood $U$ for every $g \in G$ such that~\cite[Cor.\ 5.1]{PennecLorenzi2020}
    \begin{align} 
            \Exp_g(v) &= g \exp \left(\dd_g \Lt_{g^{-1}}(v) \right) = \exp \left(\dd_g \Rt_{g^{-1}}(v) \right) g, \quad v \in T_gG,\nonumber \\
            \Log_g(f) &= \dd_e\Lt_g \big(\log(g^{-1}f)\big) = \dd_e\Rt_g \big(\log(fg^{-1})\big), \quad f \in U. \label{eq:Log}
        \end{align}
    Geodesics in $U$ are then of the form
    \begin{equation} \label{eq:geodesic}
            t \mapsto \gamma(t; g, f) = \Exp_g(t \Log_g(f)) = g \exp(t\log(g^{-1}f)) = \exp(t\log(fg^{-1}))g.
    \end{equation}
    This fact makes the CCS connection the canonical connection on $G$.

    From the above, we have two equivalent ways to parametrize geodesics in $U$: We can use boundary values $(g, f)$ or initial values $(g, v)$, which are connected by $(g, f) \leftrightarrow (g, \Log_g(f))$. We will use the former convention in this work, as its inherent symmetry simplifies, e.g., the implementation of derivatives of geodesics with respect to their parameters. Particularly, estimators of regression problems will be defined by the start and endpoint of the optimal geodesic.

    \subsubsection{Riemannian geodesic regression.} 
    Geodesic regression was proposed for Riemannian manifolds by Fletcher in~\cite{Fletcher2013}. We now summarize the approach for Lie groups, adapting our parametrization of geodesics with initial values for consistency of exposition.
    
    For Riemannian geodesic regression, the Lie group $G$ must have a Riemannian metric---a smoothly varying inner product on the tangent spaces. The metric determines the corresponding distance function $\dist$ and Levi-Civita connection of $G$; the latter defines the geodesics.
    
    For a $G$-valued random variable $Y$ and a non-random, real-valued variable $t$, the generalization of the linear regression model to the manifold setting is the geodesic model
    \begin{equation} \label{eq:model}
        Y = \Exp_{\gamma(t; g_0, g_1)}(\epsilon),
    \end{equation}
    where $\epsilon$ is a random variable taking values in the tangent space at $\gamma(t; g_0, g_1)$. The points $g_0, g_1 \in G$ are the model's parameters that must be estimated. 
    
    Let $U \subset G$ be a normal convex neighborhood. Without loss of generality, we assume in this work that the variable $t$ takes values in the interval $[0,1]$. Given data $(f_1,t_1), \dots,(f_N, t_N) \in U \times [0,1]$, the \textit{least-squares estimator} $\vartheta^{LS}: (U \times [0,1])^N \to U^2$ is the minimizer of the \textit{sum-of-squared error}
    \begin{equation} \label{def:least_squares_estimator}
        \vartheta^{LS} \left( (f_i, t_i)_{i=1}^N \right) = \argmin_{(g_0, g_1) \in U^2} E(g_0,g_1) := \argmin_{(g_0, g_1) \in U^2} \frac{1}{2}\sum_{i=1}^N \dist^2\big(\gamma(t_i; g_0, g_1), f_i \big).
    \end{equation}
    It can be computed using gradient descent. 
    Denoting the adjoint operator with respect to the Riemannian metric
    by ``$*$'', the gradients of $E$ with respect to $g_0$ and $g_1$ are given by~\cite[Section\ 3]{Hanik_ea2020}
    \begin{align}
        \grad_{g_0} E(\fcdot, g_1) &= - \sum_{i=1}^N \big(\dd_{g_0} \gamma(t_i; \fcdot, g_1) \big)^* \Log_{\gamma(t_i)}(f_i) \qquad\textnormal{and} \label{eq:gradient_p}\\
        \grad_{g_1} E(g_0, \fcdot) &=
        - \sum_{i=1}^N \big(\dd_{g_1} \gamma(t_i; g_0, \fcdot) \big)^* \Log_{\gamma(t_i)}(f_i), \label{eq:gradient_v}
    \end{align}
    respectively.\footnote{The adjoint operators in \eqref{eq:gradient_p} and \eqref{eq:gradient_v} are given in terms of so-called Jacobi fields; see~\cite[Section\ 5.1]{Hanik2023} for information on when and how they can be computed explicitly.} Thus, the desired optima in~(\ref{def:least_squares_estimator}) are roots to both \eqref{eq:gradient_p} and \eqref{eq:gradient_v}. 

    \subsubsection{The problem with Riemannian geodesic regression in Lie groups.}
    A natural property we can expect from an estimator $\vartheta: (U \times [0,1]
    )^N \to U^2$ for \eqref{eq:model} is that it is equivariant under joint translations of the data; that is, for every $h \in G$,
    \begin{align*}
        \vartheta \left( \big(\Tt_h(f_i), t_i \big)_{i=1}^N \right) &= \Tt_h \Big(\vartheta \left((f_i, t_i)_{i=1}^N \right) \Big), \qquad \Tt_h \in \{ \Lt_h, \Rt_h \}.
    \end{align*}
    Indeed, translations are the symmetries of $G$, and, intuitively, the solution to a symmetry-transformed version of the regression problem should be the symmetry-transformed original solution. 

    The following example shows that equivariance under translations trivially holds for ordinary linear regression.
    \begin{example}
        Let $M = \mathbb{R}^n$ with the Euclidean metric. It is also a Lie group with vector addition as the group operation. We have $\Lt_z(x) = \Rt_z(x) = x + z$, so the Euclidean metric is bi-invariant. Consequently, the CCS connection coincides with the Levi-Civita connection, which is just the ordinary directional derivative. With the identification $T_x\mathbb{R}^n \cong \mathbb{R}^n$, we have $\Exp_x(v) = x + v$, implying that geodesics are straight lines.
        With this, we see that \eqref{def:least_squares_estimator}) is the standard least-squares estimator of linear regression, which is equivariant under translations.
    \end{example} 
    The equivariance of the least-squares estimator for linear regression has fundamental consequences in applications. There is, e.g., the helpful effect that an experimenter does not need to consider the placement of the origin of the coordinate system; if it were not equivariant, then the quality of the result would depend on it.

    Because of the raised points, it is a major drawback that, whenever $G$ does not possess a bi-invariant metric, the Riemannian least-squares estimator~(\ref{def:least_squares_estimator}) is not equivariant under all joint translations of the data. 
    Indeed, for such $G$ either left or right translations do not preserve distances, leading to not-only-translated least-squares estimators. (See Section~\ref{subsec:simulation_study} for an example.)
    Two important examples of Lie groups without a bi-invariant metric are the general linear group $\GL(n)$ and the special Euclidean group $\SE(n)$ including the group of rigid-body motions. There are many applications where data from these groups is considered~\cite{AmbellanZachowvonTycowicz2019_GL3,Boisvert_ea2008,Cesic_ea2016,vemulapalli2014human,WU2005981}.
    In the next section, we thus propose an alternative estimator that does not suffer from this issue.

    \subsubsection{Related work and contribution}

    Regression for geometric data has gained increasing attention in recent years, driven by significant empirical improvements~\cite{adouani2024regression,Cornea_ea2017,Fletcher2013,giovanis2020data,hanik2023sundials,hanik2024casteljau} that have been obtained using intrinsic approaches for such data.
    Most approaches in this line of work employ the Riemannian framework to generalize regression methods from multivariate statistics. Beyond geodesic regression introduced by Fletcher~\cite{Fletcher2013}, there are Riemannian doppelgangers for non-parametric regression~\cite{davis2010regression,mallasto2018wrapped,Steinke_nonparameric} as well as higher-order models including polynomial~\cite{hinkle2012polynomial} and spline~\cite{adouani2024regression,Hanik_ea2020} approaches.
    However, all methods suffer from the discussed problem when data from Lie groups without a bi-invariant metric is considered.

    An incompatibility of the group and Riemannian structure was also recognized in other works, where the authors investigate more general notions of metrics. One option is to consider bi-invariant pseudo-Riemannian metrics~\cite{woods2003}; these are again only available for special cases~\cite{MiolanePennec2015}.
    Another option poses Cartan-Schouten metrics~\cite{Diatta2024}, that is, (pseudo-)Riemannian metrics whose Levi-Civita connection agrees with CCS connection, albeit at the possible expense of invariance.    
    For both options, the potential for generalizing statistical procedures under the lack of invariance or positive definiteness remains to be explored.

    Beyond the Riemannian statistical framework, a promising line of work follows a more general approach based on an affine connection structure~\cite{PennecLorenzi2020}.
    With such a non-metric structure, statistical estimators cannot be defined by optimizing residuals or variance as in Riemannian manifolds.
    However, one may obtain generalizations in terms of weaker notions.
    In particular, Pennec and Arsigny~\cite{PennecArsigny2013} characterized the mean as exponential barycenters (the critical points of the variance in the Riemannian case) allowing in turn for bi-invariant notions of higher order moments and Mahalanobis distance.
    Building on these, Hanik et al.~\cite{Hanik2022} derived dissimilarity measures for sample distributions allowing for bi-invariant two-sample tests. 
    

    In this work, we generalize linear regression by modeling the relationships between a Lie group-valued response and a scalar explanatory variable via the geodesics of the CCS connection.
    Employing the affine connection structure of the group, we thus obtain a geodesic estimator that is completely consistent with left and right translation, while reproducing the well-known least-squares approach for the Euclidean case.
    To the best of our knowledge, this is the first bi-invariant estimator for geodesic regression.
    We further propose an efficient numerical scheme for estimation that exploits a reformulation using Jacobi fields computable in a single pass of automatic differentiation.
    The source code of our prototype implementation will be published as part of the open-source library \texttt{Morphomatics}~\cite{Morphomatics}.
    
\section{Bi-invariant Geodesic Regression}
From now on, we consider $G$ endowed with the CCS connection. The geodesics in $G$ are thus given by~\eqref{eq:geodesic}.
In the following, we will need the fact that the start point map $g \mapsto \gamma(t; g, f)$ is differentiable in $U$ with an invertible derivative; that is, the inverse of the partial derivative
$\dd_g \gamma(t; \fcdot, f): T_gG \to T_{\gamma(t; g, f)}G$
exists for all $t \in [0,1]$ and $f \in U$. Because of the symmetry $\gamma(t; g, f) = \gamma(1-t; f, g)$, this is also true for the derivative with respect to the endpoint.

Our proposition for an equivariant estimator starts with \eqref{eq:gradient_p} and \eqref{eq:gradient_v} but replaces the metric-dependent adjoint with the weighted connection-dependent inverted differential. Both operators are intimately related via Jacobi fields~\cite[Sections~3-4]{BergmanGousenbourger2018}, but the inverted differential does not depend on a given metric.

\begin{definition} \label{def:bi_inv_estiamtor}
Let $U \subseteq G$ be a normal convex neighborhood. Furthermore, let $(f_1, t_1), \dots, (f_N, t_N) \in U \times [0,1]$. The \textit{bi-invariant estimator} for (\ref{eq:model}) is the geodesic between the points $\vartheta^{\textnormal{BI}}( (f_i, t_i \big)_{i=1}^N):= (\hat g_0,\hat g_1) \in U^2$ that satisfy
\begin{align}
    \sum_{i=1}^N(1-t_i)^2 \big( \dd_{\hat g_0} \gamma(t_i; \fcdot , \hat g_1) \big)^{-1} \left (\Log_{\gamma(t_i)}(f_i) \right)&= 0 \qquad\textnormal{and} \label{eq:bir1}\\
    \sum_{i =1}^N t_i^2 \big( \dd_{\hat g_1} \gamma(t_i; \hat g_0, \fcdot) \big)^{-1} \left( \Log_{\gamma(t_i)}(f_i) \right) &= 0.\label{eq:bir2}
\end{align}
\end{definition}

A generalization should always coincide with the notion it is motivated by. The following remark shows this is the case for the new estimator: It reduces to the standard least-squares estimator in the Euclidean case.
\begin{remark}
        Denoting the $n$-by-$n$ identity matrix by $I$, we find for $t \in [0,1]$ and $x,y \in \mathbb{R}^n$
        \begin{align*}
            \dd_{x_0}\gamma(t; \fcdot, x_1) &= \dd_{x_0}(z \mapsto z + t(x_1-z)) = (1-t) I \qquad\textnormal{and} \\
            \dd_{x_1}\gamma(t; x_0, \fcdot) &= \dd_{x_1}(z \mapsto x_0 + t(z-x_0)) = t I.
        \end{align*}
        Hence, since the adjoint is the matrix transpose, we find
        \begin{align*}
            \big(\dd_{x_0} \gamma(t;\fcdot,x_1)\big)^* &= (1-t) I = (1-t)^2 \big( \dd_{x_0} \gamma(t;\fcdot,x_1) \big)^{-1} \qquad\textnormal{and} \\
             \big(\dd_{x_1} \gamma(t;x_0,\fcdot)\big)^* &= t I = t^2 \big(\dd_{x_1} \gamma(t;x_0,\fcdot)\big)^{-1}.
        \end{align*}
        It follows that $\vartheta^{\textnormal{BI}} = \vartheta^{\textnormal{LS}}$ in Euclidean space $G = \mathbb{R}^n$ since we look for the (unique) zero of the same system of equations.
    \end{remark}

    We will now show that the new estimator does not suffer from the same problem as the least-squares estimator. To this end, we need the following Lemma:
\begin{lemma}\label{lem:equiv}
    Geodesics $\gamma(\fcdot; g_0, g_1)$ are equivariant under joint translations of their endpoints $g_0$ and $g_1$; that is, for all $t \in [0,1]$ and $h \in G$,
\begin{align*}
    \gamma(t;  \Tt_h(g_0, g_1)) &= \Tt_h\left( \gamma(t; g_0, g_1) \right), \qquad \Tt_h \in \{ \Lt_h, \Rt_h \}.
\end{align*}
\end{lemma}

\begin{proof} 
    Because of the characterization of
    geodesics with left translations in~\eqref{eq:geodesic}, we find
    \begin{align*}
        \gamma(t; hg_0, hg_1)
        = hg_0 \exp \left( t\log \left(g_0^{-1} g_1 \right)  \right)
        = h\gamma(t;g_0, g_1).
    \end{align*}
    The proof works analogously for right translations using the characterization of geodesics with right translations in \eqref{eq:geodesic}.
\end{proof}

    We also need the following facts:
    Clearly, $(\Lt_g)^{-1} = \Lt_{g^{-1}}$ and $(\Rt_g)^{-1} = \Rt_{g^{-1}}$; hence, for all $g,h \in G$,
    \begin{equation} \label{eq:diff_inverse}
        (\diff_g \Lt_h)^{-1} = \diff_{hg} \Lt_{h^{-1}} \qquad \textnormal{and} \qquad (\diff_g \Rt_h)^{-1} = \diff_{gh} \Rt_{h^{-1}}.
    \end{equation}
    Furthermore, translations also induce \textit{joint} translations on the product group\footnote{A product of Lie groups is also a Lie group with operations working entry-wise.} $G^m$ by $\Lt_h((g_i)_{i=1}^m) := ((\Lt_h(g_i))_{i=1}^m)$ and $\Rt_h((g_i)_{i=1}^m) := ((\Rt_h(g_i))_{i=1}^m)$, which are also automorphisms. Note that we use the same symbols $\Lt$ and $\Rt$ for translations and (real) joint translations.
    We are now ready for the following theorem:
    \begin{theorem} \label{thm:main}
    Let $(f_1, t_1), \dots, (f_N, t_N) \in U \times \mathbb{R}$. Then, for all $h \in G$,
    \begin{align*}
        \vartheta^{\textnormal{BI}} \left( \big(\Tt_h(f_i), t_i \big)_{i=1}^N \right) &= \Tt_h \Big(\vartheta^{\textnormal{BI}} \left((f_i, t_i)_{i=1}^N \right) \Big), \qquad \Tt_h \in \{ \Lt_h, \Rt_h \}.
    \end{align*} 
    \end{theorem}
\begin{proof}
    In the following, we consider the geodesic boundary map $\gamma_t: U^2 \to U$, $(g_0, g_1) \mapsto \gamma(t; g_0, g_1)$, to avoid always writing two cases. (Note that $t$ appears as a subscript in this proof as it can be interpreted as arbitrary but fixed.) Remember that $T_{(g_0,g_1)}G^2 \cong T_{g_0}G \times T_{g_1}G$. Let $s: [0,1] \times T_{(g_0,g_1)}G^2 \to T_{(g_0,g_1)}G^2, (t, (v, w)) \mapsto s_t(v, w) := ((1-t) ^2v, t^2w)$, be multiplication of tangent vectors of $G^2$ with the scalars from Definition~\ref{def:bi_inv_estiamtor}.
    
    We set $b_i=\gamma_{t_i}(g_0,g_1)$ for $i=1,\dots,n$. Our goal is to show that the vector
    \begin{equation} \label{eq:force_2}
        v \left(g_0, g_1, (f_i)_{i=1}^N \right) := \sum_{i=1}^N s_{t_i} \left( \diff_{(g_0, g_1)} \gamma_{t_i} \right)^{-1} \left( \Log_{b_i} (f_i)\right) \in T_{(g_0, g_1)}G^2
    \end{equation}
    fulfills $v \circ \Lt_h = d \Lt_h \circ v$ as a function on $U^{N+2}$.  
    The derivative $d_{(g_0, g_1)} \Lt_h$ being an invertible, linear operator (with inverse $d_{(hg_0, hg_1)} \Lt_{h^{-1}}$) at each $(g_0,g_1) \in G^2$ then implies that $v(g_0, g_1, (f_i)_{i=1}^N) = 0$ if and only if $v(hg_0, hg_1, (hf_i)_{i=1}^N) = 0$.

    From Lemma~\ref{lem:equiv}, we get for all $t$,
    $$\gamma_t \circ \Lt_h = \Lt_h \circ \gamma_t;$$
    so differentiating at $(g_0, g_1)$ and applying the chain rule\footnote{Let $M_1, M_2, M_3$ be smooth manifolds and $F: M_1 \to M_2$, $H: M_2 \to M_3$ smooth. The chain rule yields the differential $\diff_p(H \circ F) = \diff_{F(p)} H \circ \diff_pF$ where $(H \circ F)(p)$ is defined.} yields
    $$\diff_{(hg_0, hg_1)} \gamma_t \circ \diff_{(g_0, g_1)} \Lt_h = \diff_{\gamma_t(g_0, g_1)} \Lt_h \circ \diff_{(g_0, g_1)} \gamma_t.$$
    Together with~\eqref{eq:diff_inverse}, this implies
    \begin{equation} \label{eq:inv_gam_trans}
        \left( \diff_{(hg_0, hg_1)} \gamma_t \right)^{-1} = \diff_{(g_0, g_1)} \Lt_h \circ \left( \diff_{(g_0, g_1)} \gamma_t \right)^{-1} \circ \diff_{h\gamma_t(g_0, g_1)} \Lt_{h^{-1}}.
    \end{equation}
    Combining \eqref{eq:force_2}, \eqref{eq:Log}, and \eqref{eq:inv_gam_trans} gives
    \begin{align*}
        (v &\circ \Lt_h) \left(g_0, g_1, (f_i)_{i=1}^N \right) \\
        &= \sum_{i=1}^N s_{t_i} \left( \diff_{(hg_0, hg_1)} \gamma_{t_i} \right)^{-1} \left( \Log_{ hb_i } (hf_i)\right)\\
            &= \sum_{i=1}^N s_{t_i} \Big( \diff_{(g_0, g_1)} \Lt_h \circ \left( \diff_{(g_0, g_1)} \gamma_{t_i} \right)^{-1} \circ \diff_{hb_i} \Lt_{h^{-1}} \circ \diff_e \Lt_{hb_i} \Big) \Big( \log \big(b_i^{-1} f_i \big) \Big).
    \end{align*}
    We now apply the chain rule again and, pull the first differential out of the sum. The latter can be done since there is no cross-talk between the two copies of $G$ under joint translations. We thereby obtain
    \begin{align*}
        (v &\circ \Lt_h) \left(g_0, g_1, (f_i)_{i=1}^N \right) \\
        &= \diff_{(g_0, g_1)} \Lt_h \bigg( \sum_{i=1}^N \Big( s_{t_i} \left( \diff_{(g_0, g_1)} \gamma_{t_i} \right)^{-1} \circ \diff_e \left( \Lt_{h^{-1}} \circ \Lt_{hb_i} \right) \Big) \Big( \log \big(b_i^{-1} f_i \big) \Big) \bigg)\\
        &= \diff_{(g_0, g_1)} \Lt_h \bigg( \sum_{i=1}^N \Big( s_{t_i} \left( \diff_{(g_0, g_1)} \gamma_{t_i} \right)^{-1} \circ \diff_e\Lt_{b_i} \Big) \Big(\log \big(b_i^{-1} f_i \big) \Big) \bigg) \\
        &= \diff_{(g_0, g_1)} \Lt_h \bigg( \sum_{i=1}^N \Big( s_{t_i} \left( \diff_{(g_0, g_1)} \gamma_{t_i} \right)^{-1} \Big) \Big(\Log_{b_i} (f_i) \Big) \bigg) \\
        &= \left( \diff_{(g_0, g_1)} \Lt_h \circ v \right) \left(g_0, g_1, (f_i)_{i=1}^N \right),
    \end{align*}
    which is what we wanted to show.
    
    The proof for right translation works analogously by replacing left with right translations everywhere.  
\end{proof}
The theorem shows that the novel estimator is indeed equivariant under left and right translations. This and the fact that it coincides with linear regression in Euclidean space make it a natural choice for regression in Lie groups.

\section{Computation}
Note that the components of the vector $v(g_0, g_1, (f_i)_{i=1}^N)$ in the proof act like a ``net-force'', which the data points apply (through the geodesic) to $g_0$ and $g_1$. Therefore, we propose the following algorithm to approximate solutions of \eqref{eq:bir1} and \eqref{eq:bir2}: Starting with initial guesses $\hat g_0 ^{(0)}$ and $\hat g_1 ^{(0)}$, we use the iteration
\begin{align*}
    \hat g_j^{(k+1)} &:= \Exp_{\hat g_j^{(k)}} \left( \lambda v_j \right) \quad \textnormal{with} \quad (v_0, v_1) := v \left(\hat g_0^{(k)}, \hat g_1^{(k)}, (f_i)_{i=1}^N \right).
\end{align*}
with a tuned stepsize $\lambda$; convergence can be determined, e.g., by observing the length of the update vector $v$ according to some auxiliary norm. Similar algorithms have been used with great success. Indeed, the algorithm for computing the equivariant mean of Pennec and Arsigny~\cite[Algorithm~1]{PennecLorenzi2020} can be seen as a special case, in which only constant geodesics are considered.

In each step, we must compute the vectors $v_0$ and $v_1$ at the current iterate. Luckily, we need not explicitly invert the differential operators in~\eqref{eq:bir1} and \eqref{eq:bir2} for this. Consider a general geodesic $\gamma$ in $U$ to see this. When viewed as a function in $t$ and with $v \in T_{g_0}G$,
$$J: t \mapsto \dd_{g_0} \gamma \left(t; \fcdot , g_1 \right) (v)$$
is a so-called Jacobi (vector) field along the geodesic $\gamma (\fcdot; g_0, g_1)$ that fulfills the Jacobi equation---a linear second-order differential equation~\cite{BergmanGousenbourger2018}. Importantly, $J(0)=v$ and $J(1)=0$. A Jacobi field also encodes the derivative with respect to the endpoint because of the symmetry of geodesics under the exchange of the start and endpoint~\cite[Section~3.1]{BergmanGousenbourger2018}.

Let now $w:= \dd_{g_0} \gamma \left(t_0; \fcdot , g_1 \right) (v)$ for some $t_0 \in (0,1)$. Then, $v$ can be determined from $w$ by finding the Jacobi field $J$ and evaluating at $t = 0$. To this end, set $f := \gamma(t_0; g_0, g_1)$. The linear reparametrization $s(t) := \nicefrac{(1-t)}{(1-t_0)}$  yields the geodesic $\eta: [0,\nicefrac{1}{(1-t_0)}] \to U, s \mapsto \gamma(s;g_1,f)$. Note that $\eta(\nicefrac{1}{(1-t_0)}) = g_0$. The Jacobi field
$$\widetilde{J}: t \mapsto \dd_{f} \eta \left(t; g_1, \fcdot \right) (w)$$
fulfills $\widetilde{J}(0) = 0$ and $\widetilde{J}(1)=w$.
Therefore, we have two Jacobi fields $J$ and $\widetilde{J}$ along the same---albeit differently parametrized---geodesic that coincide at $f = \gamma(t_0) = \eta(1)$ and $g_1 = \gamma(1) = \eta(0)$.  

Solutions to the Jacobi equation are well-known to be equivariant under reflections and linear reparametrizations. Therefore, and
since solutions to second-order equations are uniquely determined by the values at two different points, we have $\widetilde{J}((1-t)/(1-t_0)) = J(t)$ for all $t \in [0,1]$. In particular, it follows that
$$\left(\dd_{g_0} \gamma \left(t_0; \fcdot , g_1 \right) \right)^{-1} (w) = \dd_f \eta \left(\frac{1}{1-t_0}; g_1, \fcdot \right) (w) = \dd_f \gamma \left(\frac{1}{1-t_0}; g_1, \fcdot \right) (w).$$
The right-hand side can be efficiently computed using, e.g., automatic differentiation on \eqref{eq:geodesic} and projecting the result to $T_{g_0}G$. 

For differentials with respect to endpoints, everything works analogously when replacing $(1-t)$ by $t$.

\section{Experiments}

In the following, we evaluate the proposed regression scheme experimentally both for a synthetic example and in an application to knee joint configurations during the progression of osteoarthritis.

Throughout this section, the dependent variables will be elements in the Lie group $\SE(3)$ of rigid-body motions.
We denote the group of 3-by-3 rotation matrices and their identity matrix by $\SO(3)$ and $I$, respectively. Remember that $\SE(3) = \SO(3) \ltimes \mathbb{R}^3$ is a semidirect product with the group operation
$(R, x)(Q, y) = (RQ, x+Ry)$. Its tangent space $\mathfrak{se}(3)$ at the identity element $(I, 0)$ is the product of the set of skew-symmetric 3-by-3 matrices and $\mathbb{R}^3$. The inverse of an element $(R, x)$ is given by $(R^\intercal, - R^\intercal x)$. As the name suggests, $\SE(3)$ consists of the motions a rigid body can perform in 3-space. Naturally, it appears in many real-world applications. It is well known that $\SE(3)$ does \textit{not} possess a bi-invariant metric~\cite{zefran1999metrics}. Screw motions are 1-parameter subgroups and, thus, the geodesics of the CCS connection.

\subsection{Simulation study} \label{subsec:simulation_study}

We start with a collection of synthetic experiments demonstrating the equivariance of our proposed scheme on the one hand and, on the other, juxtaposing it to the variance of estimators based on the Riemannian framework. In particular, we compare our algorithm to the Riemannian geodesic regression by Fletcher~\cite{Fletcher2013}. 
To this end, we equip the group of rigid motions with a Riemannian structure based on the product of the standard bi-invariant metrics of $\SO(3)$ and $\R^3$---a frequent choice in the literature~\cite{belta2002euclidean,zefran1999metrics}.

We simulated data according to model (\ref{eq:model}) by determining a geodesic in terms of a random start and endpoint and, subsequently, perturbing it at $N=10$ equidistant time points in the unit interval.
To generate random points and vectors we pseudo-randomly sample centered normal distributions in $\mathfrak{se}(3)$ with isotropic variance of 1 and 0.01, respectively. Based on these samples, tangent vectors are obtained via left-translation, whereas points are computed via the group exponential.
To avoid a bias towards a specific geometric structure in our comparison, we use the random samples to derive two datasets from model (\ref{eq:model}): The first ($\mathcal{D}_{\textnormal{CCS}}$) using the CCS connection for exponentiation and geodesic interpolation, while the second ($\mathcal{D}_{\textnormal{LC}}$) is based on the Levi-Civita connection.

\begin{figure}
         \centering
         \includegraphics[width=.8\linewidth]{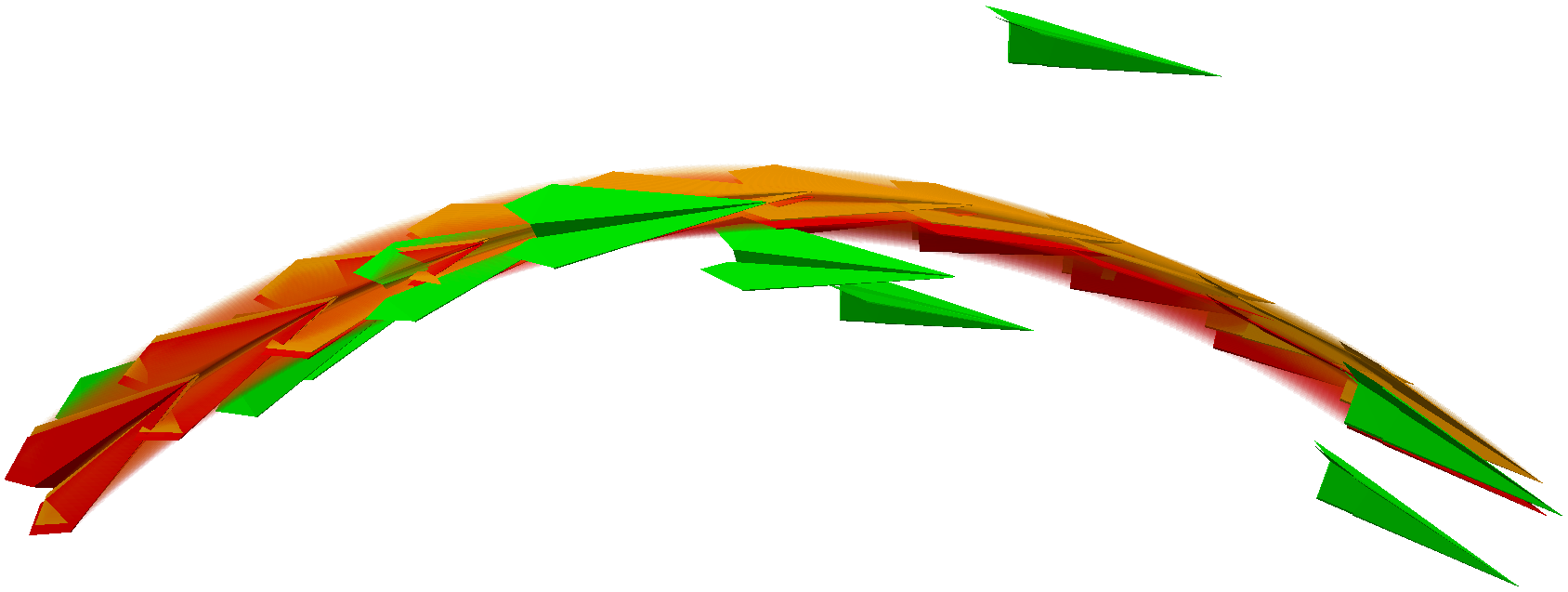}
         \caption{Equivariance of bi-invariant regression under symmetries visualized by the action of SE(3)-valued data on a reference object (paper plane).
         Green: $\Rt_h(\mathcal{D}_{\textnormal{CCS}})$, that is, sample data $\mathcal{D}_{\textnormal{CCS}}$ after right translation with random $h$; Red: Estimated geodesic by bi-invariant regression on $\Rt_h(\mathcal{D}_{\textnormal{CCS}})$; Orange: $\Rt_h(\gamma)$, that is, the geodesic $\gamma$ from bi-invariant regression on $\mathcal{D}_{\textnormal{CCS}}$ right-translated by $h$.}
         \label{fig:biinvariant_regression}
\end{figure}

In Figure~\ref{fig:biinvariant_regression}, we show estimates of the geodesic trend underlying $\mathcal{D}_{\textnormal{CCS}}$ as determined by our proposed regression scheme (with stepsize $\lambda=0.1$) after right translation by a random motion $h$.
Alongside the translated dataset $\Rt_h(\mathcal{D}_{\textnormal{CCS}})$, this figure shows the translated estimate on the untransformed data in orange together with the estimate for $\Rt_h(\mathcal{D}_{\textnormal{CCS}})$.
While both estimates show a high degree of agreement, we observe a minor difference that can be attributed to finite precision arithmetic.

Next, we examine the same construction for Riemannian geodesic regression with the dataset $\mathcal{D}_{\textnormal{LC}}$ using the same random $h$ for the right translation as in the previous example.
In Figure~\ref{fig:riemannian_regression}, we again show the estimate for the translated data and the translated estimate for the original data in red and orange, respectively.
Both curves are quite different illustrating the lack of equivariance of the Riemannian estimator.
In particular, we can observe a bending of the data and estimator under translation causing a violation of geodesicity in the Riemannian sense.
A na\"ive approach to obtain a Riemannian geodesic via translation would be to apply $\Rt_h$ to the parametrization, that is, the endpoints.
The resulting geodesic---shown in yellow---is far from the estimator for $\Rt_h(\mathcal{D}_{\textnormal{LC}})$ with a goodness of fit decreasing with the distance to the endpoints.

\begin{figure}
         \centering
         \includegraphics[width=.8\linewidth]{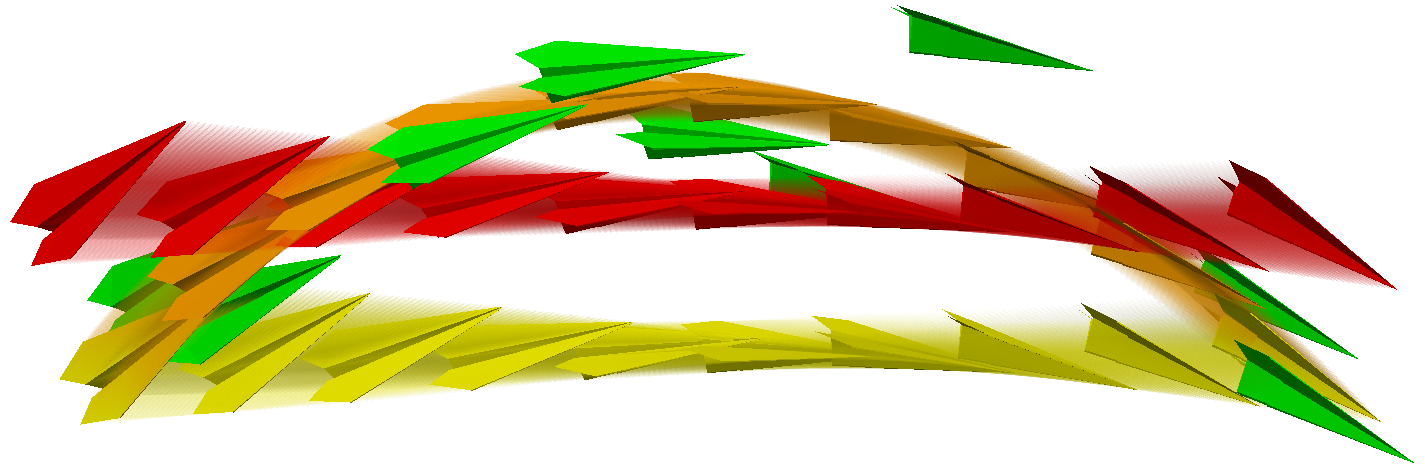}
         \caption{Variance of Riemannian regression under symmetries visualized by the action of SE(3)-valued data on a reference object (paper plane).
         Green: $\Rt_h(\mathcal{D}_{\textnormal{LC}})$, that is, sample data $\mathcal{D}_{\textnormal{LC}}$ after right translation with random $h$; Red: Estimated geodesic by Riemannian regression on $\Rt_h(\mathcal{D}_{\textnormal{LC}})$; Orange: $\Rt_h(\gamma)$, that is, the geodesic $\gamma$ from Riemannian regression on $\mathcal{D}_{\textnormal{LC}}$ right-translated by $h$; Yellow: Geodesic connecting endpoints of $\Rt_h(\gamma)$.}
         \label{fig:riemannian_regression}
\end{figure}

We further provide a quantitative assessment of the deviation from Riemannian geodesicity under the right translation and hence diminishing adequacy of the Riemannian geodesic regression.
To this end, we draw 100 random motions as before but with a variance of 100.
For each of the motions, we right-translate the dataset $\mathcal{D}_{\textnormal{LC}}$ and perform Riemannian geodesic regression.
We then assess the goodness of fit for the estimator in terms of the Riemannian coefficient of determination~\cite{Fletcher2013} denoted $R^2$.
For the untranslated dataset, the estimated Riemannian geodesic achieves a fit of $R^2 = 0.93$.
A histogram showing the distribution of this coefficient under random right translation is given in Figure~\ref{fig:R2_hist}.
Note that, while in most cases the Riemannian model can fit the translated data well (relative to the untranslated case), we can observe instances that are very poorly replicated.
This result implies that the Riemannian geodesic model can be arbitrarily unfavorable even if a reference frame exists for which the model assumption of Riemannian geodesicity applies.

\begin{figure}[tb]
    \centering
    \includegraphics[width=0.7\linewidth]{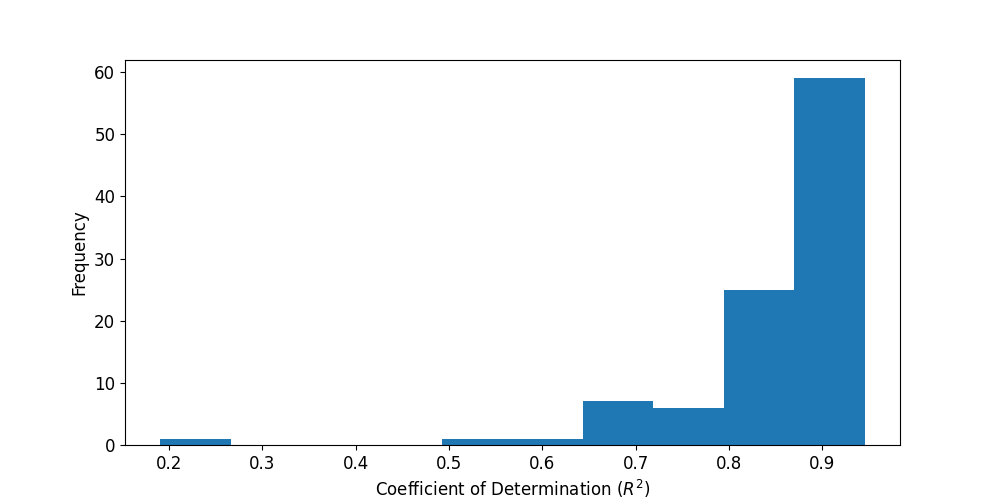}
    \caption{Histogram of $R^2$ values for Riemannian estimators under random translation.}
    \label{fig:R2_hist}
\end{figure}

\subsection{Skeletal configurations of the knee joint}
A possible application of bi-invariant geodesic regression is studying the progression of osteoarthritis (OA) in knee joints.
OA is a degenerative disease of the joints. It is characterized by a loss of cartilage leading, among others, to a narrowing of the joint space between the femur and tibia. A method of quantifying OA is the \emph{Kellgren-Lawrence (KL) grade}, which assigns integers between 0 and 4 to indicate the severity~\cite{Kellgren1957}.
In this experiment, we describe the relative position of the femur and tibia by a rigid-body transformation as in \cite[Section~4.1]{Hanik2022}. We then regress the relative positions against OA grades to (re)discover the joint space narrowing. 

Our data set stems from 50 subjects (10 per grade) randomly selected from the Osteoarthritis Initiative (a longitudinal, prospective study of knee OA) for which femur and tibia segmentations of the respective magnetic resonance images and KL grades are publicly available.\footnote{See~\cite{Ambellan_ea2019_kneeData}; the data can be found at \href{https://doi.org/10.12752/4.ATEZ.1.0}{https://doi.org/10.12752/4.ATEZ.1.0}.} From these segmentations, triangular meshes were extracted;
in a supervised post-process, the quality of segmentations was ensured.

For each femur and tibia mesh, a \emph{frame} $(R, x) \in \SE(3)$ is computed to represent its location and orientation in space: The center of gravity of the vertex set is picked as the origin $x$; the 3 positively oriented, normalized principal directions obtained by principal component analysis for the vertices are chosen as orientation. With this setup, we can express the relative position of the tibia frame $F^{(i)}_{\textnormal{T}} = (R^{(i)}_{\textnormal{T}},x^{(i)}_{\textnormal{T}})$ of the $i$-th subject to its femur frame $F^{(i)}_{\textnormal{F}} = (R^{(i)}_{\textnormal{F}}, x^{(i)}_{\textnormal{F}})$ by the rigid-body motion
$$ M^{(i)}:=F^{(i)}_{\textnormal{F}}(F^{(i)}_{\textnormal{T}})^{-1} = \left( R^{(i)}_{\textnormal{F}}(R^{(i)}_{\textnormal{T}})^\intercal, x^{(i)}_{\textnormal{F}} - R^{(i)}_{\textnormal{F}}(R^{(i)}_{\textnormal{T}})^\intercal x_{\textnormal{T}} \right).$$
The construction principle is depicted on the left of Figure~\ref{fig:knee}.

Assuming equidistant KL grades, we choose $t = 0, \nicefrac{1}{4}, \nicefrac{1}{2}, \nicefrac{3}{4}, 1$ to represent the grades 0, 1, 2, 3, and 4, respectively, and assign the corresponding value $t_i$ to each subject. We now compute $\vartheta^{\textnormal{BI}}((M_i, t_i)_{i=1}^{50})$.
To this end, the geodesic start and endpoints are initialized with a sample of a subject with KL grade 0 and 4, respectively; a stepsize $\lambda = 0.01$ is used.
Evaluating the obtained geodesic at $t = 0, \nicefrac{1}{4}, \nicefrac{1}{2}, \nicefrac{3}{4}, 1$ gives the ``denoised'' knee configurations $\hat{M}^{(j)} = (\hat{R}^{(j)}, \hat{x}^{(j)})$ for the 5 KL grades. 

\begin{figure}
    \includegraphics[width=.9\linewidth]{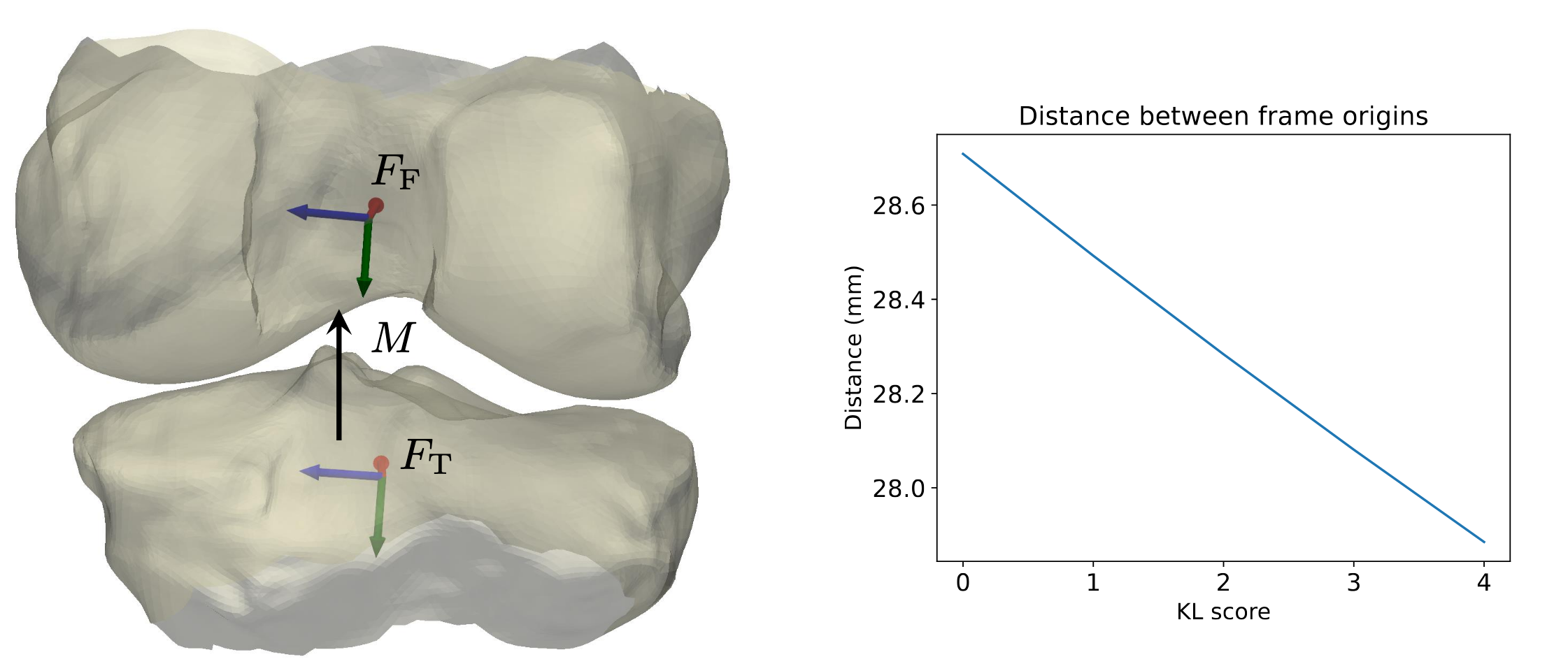}
    \caption{Left: Knee-constituting parts of the tibia (bottom) and femur (top) with their respective frames. The rigid motion $M$ moves the former's frame onto the latter's. Right: Euclidean norm of the translational part of the estimator $\vartheta^{\textnormal{BI}}((M_i, t_i)_{i=1}^{50})$ plotted against the KL grade.}
    \label{fig:knee}
\end{figure}

Using the Euclidean norm on the translations $\hat{x}^{(j)}$, we can measure the distance between the origins for the estimated geodesic; the results are shown on the right of Figure~\ref{fig:knee}.
We can see that the centers of gravity of the femur and tibia move towards each other with increasing KL grade showing the narrowing of the joint space under the progression of OA. The magnitude of the decrease in width is comparable to what is reported in the literature; see, e.g.,~\cite{RATZLAFF20181215}.\footnote{The decrease reported in~\cite{RATZLAFF20181215} is slightly larger. But their data comes from radiographic images usually taken from standing subjects, whose joints were strained. Our MRI data comes from lying subjects, whose joints were relaxed.}

\begin{credits}
\subsubsection{\ackname} This work was supported by the Federal German Ministry for Education and Research (BMBF) under Grants BIFOLD24B, 01IS14013A-E, 01GQ1115, 01GQ0850, and 031L0207D. We are grateful for the open-access data set of the Osteoarthritis Initiative, which is a public-private partnership comprised of five contracts (N01-AR-2-2258; N01-AR-2-2259; N01-AR-2-2260; N01-AR-2-2261; N01-AR-2-2262) funded by the National Institutes of Health, a branch of the Department of Health and Human Services, and conducted by the OAI Study Investigators. Private funding partners include Merck Research Laboratories; Novartis Pharmaceuticals Corporation, GlaxoSmithKline; and Pfizer, Inc. Private sector funding for the OAI is managed by the Foundation for the National Institutes of Health. This manuscript was prepared using an OAI public use data set and does not necessarily reflect the opinions or views of the OAI investigators, the NIH, or the private funding partners.

\subsubsection{\discintname}

The authors have no competing interests to declare that are
relevant to the content of this article.
\end{credits}

%
%
%
\bibliographystyle{splncs04}
\bibliography{bibliography}
    
\end{document}